\documentclass[a4paper,10pt]{article}
\usepackage{amsmath,amssymb,graphicx}

\usepackage[pdftex]{hyperref}

\renewcommand{\title}[1]{\vspace{\fill}
\eject\addtolength{\baselineskip}{4pt}
{\bfseries\LARGE #1}\\[3mm]\addtolength{\baselineskip}{-4pt}}
\renewcommand{\author}[3]{\parbox[t]{75mm}
{\begin{center}{\scshape #1}\\[3mm] #2\\
 {\ttfamily #3} \end{center}}}
\newtheorem{thm}{\bfseries Theorem}
\newtheorem{lem}[thm]{\bfseries Lemma}        
\newtheorem{remark}[thm]{\bfseries Remark}    
\newtheorem{prop}[thm]{\bfseries Proposition} 
\newtheorem{cor}[thm]{\bfseries Corollary}     
\newtheorem{defn}[thm]{\bfseries Definition}

\newenvironment{proof}{\medskip                    
\noindent{\scshape Proof:}}{\quad $\Box$\medskip}  

\def\qed{
\hfill $\Box$}

\DeclareMathOperator{\bp}{bp}
\DeclareMathOperator{\maxdc}{maxdc}

\begin{document}

\begin{center}

\title{Notes on dual-critical graphs} 
\author{Zolt\'an Kir\'aly\footnotemark[1]  
}{
Department of Computer Science and \\
 Egerv\'ary Research Group (MTA-ELTE)    \\
E\"otv\"os University \\
P\'azm\'any P\'eter s\'et\'any 1/C, Budapest, Hungary
}{
kiraly@cs.elte.hu
}\footnotetext[1]{Research was supported by 
grants (no.\ CNK 77780 and no.\ CK 80124) from the National Development
 Agency of Hungary, based on a source from the Research and Technology 
Innovation Fund.}
\author{
S\'andor Kisfaludi-Bak
}{
Department of Computer Science \\
E\"otv\"os University \\
P\'azm\'any P\'eter s\'et\'any 1/C, Budapest, Hungary
}{
kbsandor@cs.elte.hu
}


\end{center}


\begin{quote}
{\bfseries Abstract:}
  We define dual-critical graphs as graphs having an acyclic
  orientation, where the indegrees are odd except for the unique source. We have
  very limited knowledge about the complexity of dual-criticality testing. 
  By the definition the problem is in NP, and a result of Bal\'azs
  and Christian Szegedy \cite{szegedyszegedy} provides a randomized polynomial
  algorithm,
  which relies on formal matrix rank computing. It is unknown whether
  dual-criticality test can be done in deterministic polynomial
  time. Moreover, the question of being in co-NP is also open.

  We give equivalent descriptions for dual-critical graphs in the general case, and
  further equivalent descriptions in the special cases of planar graphs and 3-regular
  graphs. These descriptions provide polynomial algorithms for these special classes.
  We also give an FPT algorithm for a relaxed version of dual-criticality called
  $k$-dual-criticality.
\end{quote}

\begin{quote}
\end{quote}
\vspace{3mm}

\section{Introduction} 

The name and definition of dual-critical graphs was introduced
by Andr\'as Frank \cite{PersComm} based on the paper of Szegedy and
Szegedy \cite{szegedyszegedy}. He showed that this class is very
interesting due to the following reason. Let $G(V,E)$ be a simple
graph. We denote by $e_G(X)$ the number of incident edges for any
set $X\subseteq V$. The function $e_G$ defines a polymatroid $\mathcal{P}_G$
on $V$. A graph is \emph{dual-critical}, if and only if $\mathcal{P}_G$ has
a vertex $x$ such that all coordinates of $x$ are odd.

An orientation $\overrightarrow{G}=(V,\overrightarrow{E})$ of a graph
$G$ is called \emph{acyclic} if it does not have directed cycles. A graph
$G$ is \emph{dual-critical} if it has an acyclic orientation such that all
vertices except one have an odd indegree.

This definition might look factitious, but as we will see later,
when looking for graphs with certain parity constrained acyclic orientations,
it is always possible to reduce the problem to the dual-criticality of a
slightly altered graph.

In every acyclic orientation there is a source vertex $v$. It has an even
indegree (0). Thus an orientation of a dual critical graph that satisfies
the above conditions has exactly one vertex (the source) that has an even
indegree, and $|V(G)|-1$ vertices that have odd indegrees.
Consequently a dual-critical graph is always loopless and connected. 

It is a known fact that acyclic graphs have a \emph{topological ordering}. In a
topological ordering the source vertex comes first. The orientation of the
edges is determined by the order of their endpoints: the source of an arc
always precedes its target. Let us take a topological ordering of an
orientation described above. Beginning with the second vertex, every vertex
has an odd number of predecessors to which it is connected. Consequently, the
class of dual-critical graphs can be characterized as the graphs that can be
built by taking a single vertex, and adding new vertices connected to the
previous ones by an odd number of edges. This description shows that the problem
is in NP. Such an ordering will be called a \emph{good ordering}, and the
orientation defined by a good ordering will be called a \emph{good orientation}.

A directed graph is \emph{rooted connected} if it has a vertex $r$ called root,
from which there is a path to any other vertex.

\begin{remark} \label{rooted conn} It is easy to see that a good orientation
  of a dual-critical graph is rooted connected, with the source vertex as
  root. Indeed, except the first vertex (the source or the root) in the good
  ordering, every vertex has at least one incoming arc. So one could construct
  a backward path from any vertex to the root by going backwards on incoming
  arcs.
\end{remark}

Organization of the paper: The second section introduces the basic properties of dual-critical graphs. We also present a related class called super-dual-critical graphs. In Section \ref{sec_4} we deal with the background of the terminology, which lies in planar dual-critical graphs. Section \ref{sec_5} describes the randomized algorithm of Bal\'azs Szegedy and Christian Szegedy. The next section deals with 3-regular graphs. The main theorem of this section shows that dual-critical graphs coincide with many graph classes when restricted to 3-regular graphs. One of these classes is upper-embeddable graphs, providing us a deterministic polynomial algorithm for the 3-regular case. In the final section we define $k$-dual-criticality. Andr\'as Frank asked the original question whether an FPT algorithm for testing $k$-dual-criticality can be given. We present the first FPT algorithm for this problem.

With the exception of Section 7, many proofs are omitted, but all of these proofs can be found in \cite{techreport}.

\section{Basic properties of dual-critical graphs}

\textbf{Notation:} For a subset $X$ of vertices, $i(X)$ denotes the number of
edges induced by $X$. The symmetric difference of $X$ and $Y$ is denoted by
$X\oplus Y$.  When writing congruences, the notation of the modulus will be
omitted if it is $2$, e.g., $a \equiv b$ means that
$a$ and $b$ have the same parity.

If $|V| \not \equiv |E|$ holds for a graph $G(V,E)$, we say that $G$ has \textit{good parity}, otherwise, if $|V| \equiv |E|$, we say that $G$ has \textit{bad parity}.


\begin{defn}[$T$-odd]
  Let $T\subseteq V$. An orientation of a graph $G=(V,E)$ is $T$-odd if all
  vertices in $T$ have odd indegree, and all vertices in $V-T$ have even
  indegree.
\end{defn}

\begin{thm} \label{dualkrit ekv}
The following statements are equivalent for any graph $G=(V,E)$:
\begin{enumerate}
\item[(1)] $G$ is dual-critical.
\item[(2)] For any given $v \in V(G)$ there is an acyclic orientation, in
  which all the indegrees are odd except for $v$.
\item[(3)] The graph has good parity and for every set $T\subsetneq V$ with
  $|T|\equiv|E|$, there exists a $T$-odd acyclic orientation.
\item[(4)] Either $G$ is the graph consisting of one vertex, or it has a
  two-class partition $V=A \cup B$, such that $G[A]$ and $G[B]$ are
  dual-critical, and the cut $E(A,B)$ defined by $A$ and $B$ has an odd number
  of edges (i.e., $d(A,B)=|E(A,B)|$ is odd).
\end{enumerate}
\end{thm}

\noindent{\scshape Proof:}
\begin{description}
\item[(1)$\Rightarrow$(4)] Take a good ordering of $G$. Let $w$ be the last
  vertex in the ordering. Choose $A=V(G)-w$, $B=\{w\}$.

\item[(4)$\Rightarrow$(3)] We use induction on the number of vertices. If $|T
  \cap A| \not\equiv i(A)$ and $|T \cap B| \not\equiv i(B)$, then $|T|
  \not\equiv i(A)+i(B)+d(A,B)=|E|$, a contradiction. Wlog.\ one can suppose
  $i(A) \equiv |T \cap A|$. By induction, as $G[A]$ is dual-critical, we can
  take an acyclic orientation of $G[A]$ in which the vertices of $T \cap A$
  have an odd indegree, and the vertices of $A-T$ have an even
  indegree. Direct the edges of $E(A,B)$ towards $B$.

  As $|T \cap A| \equiv i(A)$ we have that $|T \cap B| \not\equiv i(B)$. Let
  $Z$ be the set of vertices in $B$ that have an odd number of incoming edges
  from $A$, and let $T'=(T\cap B)\oplus Z$.  As $|Z|$ is odd, $|T'| \equiv
  i(B)$.

  Now one can use induction for $G[B]$ and $T'$, and fix an acyclic $T'$-odd
  orientation of $G[B]$. It is easy to check that the resulting orientation of
  $G$ is also acyclic and moreover it is $T$-odd.

\item[(3)$\Rightarrow$(2)] We may use (3) for $T=V(G)-v$.

\item[(2)$\Rightarrow$(1)] Obvious. \quad $\Box$
\end{description}

\begin{prop}\label{makesimple}
  The following operations do not change dual-criticality, i.e., a graph is
  dual-critical, if and only if using any of these operations results in a
  dual-critical graph:
\begin{enumerate}
\item[(1)] Deletion of two parallel edges,
\item[(2)] Insertion of two parallel edges between two arbitrary vertices,
\item[(3)] Division of an edge by adding a vertex in the middle,
\item[(4)] Contraction of an edge that has an endvertex with degree 2.
\end{enumerate}
\end{prop}

%

Using the operations (1) and (3) from Proposition \ref{makesimple} one can make any
graph simple by dividing loops by two vertices (so a triangle is made) and
eliminating parallel edge pairs. We might cut a connected graph this way,
but in that case Proposition \ref{makesimple} states that the original graph
was not dual-critical.

\begin{defn}[Super-dual-critical]
  A graph is called super-dual-critical, if for any vertex $v \in V(G)$, the
  graph $G-v$ is dual-critical.
\end{defn}

\begin{prop}\label{sdcprop1}
  In super-dual-critical graph either all vertex degrees are odd, or all vertex
  degrees are even. The degree parity is the same as the parity of
  $|E(G)|-|V(G)|$.
\end{prop}

\begin{proof} 
For an arbitrary vertex $v$ the graph $G-v$ has good parity, thus
\begin{equation}
|V(G)|-1\not\equiv|E(G)|-d(v) \, \Rightarrow \, d(v) \equiv |E(G)|-|V(G)|. 
\end{equation}
\end{proof}

\begin{cor}
  A super-dual-critical graph is dual-critical if and only if it has good
  parity, or equivalently, a super-dual-critical graph is dual-critical if and
  only if every vertex has odd degree.
\end{cor}

\begin{proof} Let $G$ be a super-dual-critical graph. If it has bad parity, then by
Proposition \ref{sdcprop1} all degrees are even, so it cannot be
dual-critical. If $G$ has good parity, then all degrees are odd. Delete an
arbitrary vertex $v$. The graph $G-v$ is dual-critical, hence $G$ is
dual-critical as well, since it can be obtained from $G-v$ by adding a vertex
that has odd degree.  \end{proof}

\begin{prop}
  The graph $G$ has a $T$-odd acyclic orientation for every $T \subsetneq
  V(G)$ for which $|T|\equiv|E(G)|$ if and only if $G$ is dual-critical or $G$
  is super-dual-critical.
\end{prop}

Finally, the following is a proposition which shows that a question about
acyclic orientations with parity constraints can be viewed as dual-criticality
of a slightly changed graph.

\begin{prop}[Be\'ata Faller]
  A graph $G=(V,E)$ has a $T$-odd acyclic orientation for an arbitrary $T
  \subseteq V$, if and only if the graph $G'$ obtained by adding a vertex $v$
  and connecting it to all vertices in $V-T$ is dual-critical.
\end{prop}

\begin{proof} If $G$ has such an orientation, then directing the edges away from $v$
will give a good orientation for $G'$. If $G'$ is dual-critical, then it has a
good orientation in which the source vertex is $v$. This orientation is
$T$-odd and acyclic in $G$.  \end{proof}

\section{Planar case, motivations}\label{sec_4}

Before talking about dual-critical graphs in more detail, some background
information should be provided about the term 'dual-critical'. In this section
we are going to use matroids. Ample introductory and advanced material can be
found on them. See e.g.,  \cite{Oxley}. In this section we also omit some basic 
definitions, see \cite{techreport}, and graphs may have parallel edges and loops.

There is a well-known result about factor-critical graphs and ear
decompositions.

\begin{thm}\cite{faktorkrit ekv fulfelbont}
  A graph is factor-critical if and only if it has an odd ear decomposition,
  i.e., an ear decomposition in which all ears have an odd number of edges.
\end{thm}

We will use the following variant of this theorem.

\begin{thm}\label{faktorkritbuild}
  A graph is factor-critical if and only if it can be built from a single
  vertex using the following operations:
\begin{itemize}
\item Addition of an edge between two vertices or addition of a loop,
\item Division of an edge into path of length three using two new vertices.
\end{itemize}
\end{thm}

\begin{thm} \label{fcdualjadc} The planar dual of a planar factor-critical
  graph $G$ is always dual-critical. (So if there are multiple dual graphs
  depending on the planar embedding of $G$, then all of them are
  dual-critical.)
\end{thm}



\begin{prop} \label{szabdbefujni} Let $G$ be a factor-critical graph. The
  graph $G'$ that is obtained from $G$ by blowing an odd cycle into a vertex
  $v$ is factor-critical.
\end{prop}

\begin{thm} \label{dcdualjafc} The dual of a planar dual-critical graph $G$
  is always factor-critical. (If there are multiple dual graphs depending
  on the planar embedding of $G$, then all of them are factor-critical.)
\end{thm}


Since dualization and factor-criticality test can be done in polynomial time,
we arrive at the following corollary.

\begin{cor}
  There is a polynomial time algorithm for deciding dual-criticality if the
  graph is planar. \qed
\end{cor}

Theorems \ref{fcdualjadc} and \ref{dcdualjafc} show that
dual-criticality and factor-criticality are dual concepts. 
Our next goal is to show that Theorems \ref{fcdualjadc} and
\ref{dcdualjafc} can be generalized.

\begin{remark}\label{2offaktorkritbuild}
  It is known that a graph is 2-connected if and only if it has an open
  ear-decomposition, i.e., an ear decomposition where we begin with a cycle,
  and the two ends of an ear cannot coincide. This statement can be ported for
  factor-critical graphs. A 2-connected graph is factor-critical if and only
  if it can be obtained from an odd cycle by adding odd length open ears.
\end{remark}

Now we state a well-known result from matroid theory. A proof can be found in
section 2.3 in \cite{Oxley}.

\begin{prop} \label{dualdualja} The graphic matroid of the dual of a planar
  graph is isomorphic to the cographic matroid of the graph. (Equivalently:
  the dual graph's graphic matroid is the dual of the graph's graphic
  matroid.)
\end{prop}

\begin{defn}\cite{szegedyszegedy}
  A sequence of circuits $\{C_0,C_1, \dots C_k\}$ of the matroid
  $M=(S,\mathcal{F})$ is called an ear-decomposition if
\begin{enumerate}
\item[(1)] $C_i - (\bigcup_{j=0}^{i-1} C_j)$ is not empty for all $1 \leq i
  \leq k$
\item[(2)] $C_i \cap (\bigcup_{j=0}^{i-1} C_j)$ is not empty for all $1 \leq i
  \leq k$
\item[(3)] $C_i - (\bigcup_{j=0}^{i-1} C_j)$ is a circuit in
  $M/(\bigcup_{j=0}^{i-1} C_j)$ for all $1 \leq i \leq k$
\item[(4)] $\bigcup_{i=0}^k C_i=S$
\end{enumerate}
An ear is a set $C_i - (\bigcup_{j=0}^{i-1} C_j)$.
\end{defn}

This definition is the matroid equivalent of the open-ended ear-decomposition
which is described in Remark \ref{2offaktorkritbuild}. It follows that a
2-connected graph is factor-critical if and only if its graphic matroid has an
odd ear-decomposition.

We need two basic lemmas from matroid theory. The notation $M^*$ indicates the
dual matroid of $M$, and $M/Z$ is used for the contraction of the subset $Z$.

\begin{lem}(Theorem 8.3 in \cite{lawler}) Let $M$ be a matroid on the set
  $S$. Then for any $Z\subseteq S$ the following holds:
\begin{equation}
(M/Z)^*=M^*-Z \text{ and } M^*/Z=(M-Z)^*
\end{equation}
\end{lem}

\begin{lem}[Proposition 2.3.1 in \cite{Oxley}] \label{propcutlemma} Let
  $M=(E(G),\mathcal{F})$ be the graphic matroid of $G$. The set $Z\subseteq
  E(G)$ is a cycle in $M^*$ if and only if it is a proper cut in $G$.
\end{lem}


%

\begin{prop}
  A 2-connected graph is dual-critical if and only if its cographic matroid
  has an odd ear decomposition.
\end{prop}

\section{A randomized algorithm by Bal\'azs and
Christian Szegedy} \label{sec_5}

\begin{defn}[\cite{szegedyszegedy}]
  Let $M$ be a connected, bridgeless matroid. We denote by $ \varphi(M)$ the
  minimal possible value of the number of even ears in an ear-decomposition of
  $M$. If $M$ is bridgeless but not connected, we define $\varphi(M)$ to be
  the sum of $\varphi(K)$ over all blocks $K$ of $M$. In particular $\varphi(M
  ) = 0$ if and only if every block of $M$ has an odd ear-decomposition.
\end{defn}

\begin{thm}[Szegedy--Szegedy, Theorem 10.8 in \cite{szegedyszegedy}]
  Let $M$ be a matroid that is representable over a field of characteristic
  $2$. There is a randomized polynomial algorithm which computes $\varphi(M)$.
\end{thm}

This gives a randomized polynomial algorithm for deciding dual-criticality. We can represent the cographic matroid of graphs over a field of
characteristic two: a graph is dual-critical if and only if
$\varphi(M^*(G))=0$.

We would like to outline this algorithm for cographic matroids. Let $T$ be the
edge set of a spanning tree of our graph. We associate independent
indeterminates with each edge of $T$: $x_e$ for all $e \in T$. The tree edges
of the fundamental cycle of $i \in E(G)-T$ will be denoted by $T_i$. Let
$A=(a_{i,j})$ be the following matrix ($i \in E(G)-T$ and $j \in E(G)-T$):
\begin{equation*}
a_{ij}= \sum_{e \in T_i \cap T_j} x_e.
\end{equation*}
If the fundamental cycles have no common edge (the sum is empty), then the
matrix entry is $0$. The corank of $A$ is equal to $\varphi(M^*)$ by the
theorem of Szegedy--Szegedy (\cite{szegedyszegedy}). So $G$ is dual-critical
if and only if $det(A)$ is not the constant zero polynomial. (The determinant is defined over the representing field.)

After choosing a large enough field of characteristic 2, this can be decided using the Schwarz--Zippel lemma (\cite{schwarzzippel}),
which provides a polynomial randomized algorithm.

\section{3-regular dual-critical graphs}

\subsection{Equivalent descriptions}

A 3-regular graph on $n$ vertices have $\frac{3n}{2}$ edges, so $n$ is
even. If the graph is also dual-critical, then it must have a good parity,
thus $n=4k+2$ for some integer $k$.


\begin{thm}\label{3-reg ekv}\footnote{Be\'ata Faller and Ervin Gy\H ori had
    the original idea that dual-criticality is equivalent to description (2)
    in the 3-regular case.}  The following are equivalent for any 3-regular
  graph $G=(V,E)$ which has $4k+2$ vertices.
\begin{enumerate}
\item[(1)] $G$ is dual-critical.
\item[(2)] There are $k+1$ independent vertices, such that their deletion
  leaves a connected graph.
\item[(3)] There are $k+1$ vertices whose deletion leaves a forest.
\item[(4)] There are some independent vertices whose deletion leaves a tree.
\item[(5)] There is a spanning tree, for which the deletion of the tree's
  edges makes a graph in which every component has an even number of edges.
\item[(6)] There is an $r$-rooted connected orientation for every vertex $r
  \in V$, in which all vertices but $r$ have an odd indegree.
\item[(7)] For every partition $\mathcal{P}$ of $V$
\begin{equation} \label{Nebesky ineq}
e(\mathcal{P})\geq |\mathcal{P}|+\bp(\mathcal{P})-1
\end{equation}
holds, where $e(\mathcal{P})$ denotes the number of edges between different
classes of $\mathcal{P}$ and $\bp(\mathcal{P})$ denotes the number of classes
in $\mathcal{P}$ spanning a subgraph which has bad parity.
\item[(8)] $G$ is upper-embeddable. \footnote{For the definition of upper-embeddability see \cite{techreport}}
\end{enumerate}
\end{thm}

\goodbreak

\begin{thm}[Furst, Gross, McGeoch \cite{FurstGrossMcGeoch}]
  There is a polynomial algorithm that decides whether a graph is
  upper-embeddable or not. It runs in $O(end\log^6 n)$ time where $e$, $n$ and
  $d$ denote the number of edges, the number of vertices and the maximum
  degree respectively.
\end{thm}

\begin{cor}
  There is an algorithm for deciding dual-criticality in the 3-regular case
  which runs in $O(n^2\log^6 n)$ time.
\end{cor}

\section{A relaxed version of dual-criticality -- k-dual-critical graphs}

\begin{defn}[k-dual-critical, \cite{PersComm}]
A graph $G(V,E)$ is called $k$-dual-critical if $V$ has a partition $\mathcal{P}$ into $k$ non-empty subsets such that the contraction of all partition classes results in a dual-critical graph. (In this contraction the loops are deleted, but the parallel edges are preserved.)
\end{defn}

Equivalently, a graph is $k$-dual critical, if and only if its vertex set
has a partition $\mathcal{P}=\{P_1,P_2, \dots P_k\}$ such that
$d(\cup_{j=1}^{i-1}P_j,P_i)$ is odd for each $2\leq i \leq k$, (recall that
$d(A,B)$ denotes the number of edges between two disjoint vertex sets $A, B$).
We say that a vertex partition is a good $k$-partition if it satisfies the
condition in the above definition. We use the term good ordering for partition
classes of a good $k$-partition as we did for vertices of a dual-critical graph.

From the definition it is trivial that a graph is dual-critical if and only if
it is $|V|$-dual-critical. It is also easy to observe that every graph is
$1$-dual-critical. Moreover, if $G$ is $k$-dual-critical, then it is
$\ell$-dual-critical for each $1 \leq \ell \leq k$, since one can take a good
ordering of the partition classes of a good $k$-partition, and unify the first
$k-\ell$ classes to get a good $\ell$-partition.

We denote by $\maxdc(G)$ the biggest number $k$ for which $G$ is
$k$-dual-critical. It is easy to verify that $G$ is Eulerian if and only if
$\maxdc(G)=1$. \textit{Optimal} partition of a graph will refer to a good
$\maxdc(G)$-partition. The following proposition is easy to verify:

\begin{prop}
The minimal number of vertex pairs that need to be contracted in $G$ to
get a dual-critical graph is $|V(G)|-\maxdc(G)$.
\end{prop}

\begin{prop}\label{kdcbasicfacts}
Let $G$ be any simple graph. If the graph spanned by a partition class in
a good $\ell$-partition is non-Eulerian, then this class can be divided into
two subclasses to get a good $(\ell\! + \!1)$-partition. Consequently in an
optimal partition of $G$ every class spans an Eulerian subgraph.
\end{prop}

\begin{proof}
Suppose that a class $P_i$ of a good partition $P_1,P_2, \dots, P_\ell$ spans
a non-Eulerian subgraph. This graph has an odd proper cut with sides $P'_i$
and $P''_i$. Since $d(P_i, \cup_{j=1}^{i-1} P_j)$ is odd, one of
$d(P'_i, \cup_{j=1}^{i-1} P_j)$ and $d(P''_i, \cup_{j=1}^{i-1} P_j)$ is odd,
and the other one is even. Suppose $d(P'_i, \cup_{j=1}^{i-1} P_j)$ is odd. Then
the partition $(P_1,P_2,\dots, P_{i-1},P'_i,P''_i,P_{i+1},\dots, P_\ell)$ is a
good $(\ell\!+\!1)$-partition of $G$.
\end{proof}

\begin{thm}\label{leftalign}
Every graph has an optimal partition where all classes contain
1 or 2 vertices except possibly the first class.
\end{thm}

\begin{proof}
Let $P_1,P_2, \dots P_k$ be a good ordering of an optimal partition
$\mathcal{P}$, and let $Q_i=\cup_{j=1}^{i} P_j$.

Starting from $P_k$ we will show that if $|P_i|>2$ then some vertices
can be moved from $P_i$ to $P_{i-1}$ (for each $i=k, k-1, \dots 2$)
leaving at most 2 vertices in $P_i$. Note that this operation preserves the
parity of $d(Q_{j-1},P_j)$ if $j\not\in \{i-1,i\}$. We fix an $i\geq 3$ and
observe that a set of vertices $S \subsetneq P_i$ can be moved to $P_{i-1}$ if
and only if
\begin{align}
d(S, P_i- S) &\equiv d(S,P_{i-1}) \label{goodmove1}\\ 
d(S,Q_{i-2}) &\equiv 0. \label{goodmove2}
\end{align}
Since $\mathcal{P}$ is optimal, by Proposition \ref{kdcbasicfacts} the
graph spanned by $P_i$ is Eulerian, thus all its cuts are even; so our
conditions for a set $S$ to be movable are $d(S,P_{i-1}) \equiv d(S,Q_{i-2})
\equiv 0$.

It is sufficient to show that if $|P_i|\geq 3$ then there is a non-empty
movable $S \subsetneq P_i$. For a vertex set $R\subsetneq P_i$ let $z_R$
be a 2-dimensional vector over the two element field representing the parity
of $d(Q_{i-2},R)$ and $d(P_{i-1},R)$. Note that $S$ is movable if and only if
$z_S=(0,0)$.

Let $A,B,C$ be any $3$-partition of $P_i$ where all the classes are non-empty.
(Such a partition exists because $|P_i|\geq 3$.) If one of $z_A,z_B$ or $z_C$
is $(0,0)$, then the corresponding set is movable. Suppose there are two equal
vectors among $z_A,z_B$ and $z_C$, eg. $z_A=z_B$. Then $z_A+z_B=(0,0)$, thus
$S=A \cup B$ is movable. The only remaining case is that the three vectors are
$(0,1), (1,1)$ and $(1,0)$ in some order. But this means that $z_A+z_B+z_C=(0,0)$,
thus $d(P_i,Q_{i-1})\equiv 0$, a contradiction.

If $i=2$, then let $v$ be any vertex in $P_2$ for which $d(P_1,v)$ is odd.
(Such vertex exists because $d(P_1,P_2)$ is odd.) The set $P_2-v$ is moveable,
since $d(P_2-v,P_1)\equiv 0$.
\end{proof}

We call a good k-partition \textit{left-aligned} if all its classes except
possibly the first one have either 1 or 2 vertices. For a constant $k$, the
$k$-dual-criticality of a graph can be determined in polynomial time. Consider
the following algorithm that decides k-dual-criticality:


\smallskip
\noindent {\large \verb|Recursive k-dual-critical|} algorithm:\\
If $k\leq 1$, then we return true. For each odd degree vertex $v$ if
\verb|Recursive k-dual-critical|$(k-1,\; G-v)$ is true, then we return
true. Then we take each vertex pair $v,w$ such that $d(v)+d(w) \equiv 1$,
and we return true if \verb|Recursive k-dual-critical|$(k-1,\; G-v-w)$ is
true. Finally, if we did not return yet, we return false.

By induction it is a routine to prove that the total time required for
this procedure is at most $O(n^{2k})$. 
%

\begin{prop} \label{maxdc(G)=2}
$\maxdc(G)=2$, if and only if $G$ is a graph that consists of an even clique
and some isolated vertices.
\end{prop}

\begin{proof}
Let $P_1,P_2$ be a left-aligned optimal partition. If $P_2$ has two
vertices, then one of them has an even number of incoming edges from $P_1$
(since $d(P_1,P_2)$ is odd). We can put that vertex into $P_1$ to get a
partition $P_1=V-\{v$\} and  $P_2=\{v\}$. We denote by $N(v)$ the neighbours
of $v$. The following are easy to verify using the fact that $G[P_1]$ is Eulerian.
Suppose $u_1$ and $u_2$ are disconnected vertices in $N(v)$. The partition
$(P_1- \{u_1,u_2\}, \{v\}, \{u_1\}, \{u_2\})$ is a good 4-partition of $G$,
a contradiction. If $u_1\in N(v)$ and $u_2\in P_1 - N(v)$ are connected,
then $(P_1- \{u_1,u_2\}, \{v,u_2\}, \{u_1\})$ is a good 3-partition of $G$,
a contradiction. If $u_1,u_2 \in P_1- N(v)$ are connected, then the partition
$(P_1- \{u_1,u_2\}, \{u_1\}, \{v,u_2\})$ is a good 3-partition of $G$, a
contradiction.
\end{proof}

We call a partition \textit{maximal} if all its classes are Eulerian.
\medskip

\noindent {\large \verb|Greedy algorithm|} for constructing maximal partitions:
\smallskip

Using the idea from the \verb|Recursive k-dual-critical| method we can construct
a maximal left-aligned partition of any graph $G$. If there is an odd degree
vertex $v$ or disconnected vertex pair $v,w$ with $d(v)+d(w)\equiv 1$ whose
deletion \textit{leaves a non-Eulerian graph}, then let this vertex or vertex
pair be the last class of the partition, and delete this class from $G$. We
repeat this step until no such class can be formed. At this point the graph is
still non-Eulerian, but the deletion of any odd degree vertex leaves an Eulerian graph. 

By the proof of Proposition \ref{maxdc(G)=2} it is easy to see that the remaining
graph consists of an even clique and some isolated vertices. We choose an odd
degree vertex $v$, and add a new partition class $\{v\}$. The rest of the points
will provide the first class of the partition. The partition is maximal, since
every class spans an Eulerian graph (we made sure to put disconnected vertex pairs
in the 2-vertex classes).

By continuously updating vertex degrees, this algorithm runs in $O(kn^2)$ time. 

\medskip
%

Let $\mathcal{P}$ be a fixed maximal left-aligned partition with $\ell$ classes. We show further possible ways of increasing the number of classes in a maximal partition. Let $\sigma_0:P_1 \rightarrow \{2,3, \dots \ell, \infty\}$  denote the smallest index $i$ for which $d(v,P_i)\equiv 0$ for $v\in P_1$. If there is no such index, then we define $\sigma_0(v)=\infty$. We define $\sigma_1:P_1 \rightarrow \{2,3, \dots \ell, \infty\}$ similarly as the smallest index $i$ for which $d(v,P_i)\equiv 1$.

In the following four cases we will be able to increase the number of classes.
\begin{description}\label{opspage}

\item[(1)] Suppose that $u$ and $v$ are distinct isolated vertices in $G[P_1]$ and $\sigma_1(u)=\sigma_1(v)=s<\infty$.

\item[(2)] Suppose that $u$ and $v$ are distinct isolated vertices in $G[P_1]$ and $\sigma_1(u)=\sigma_1(v)=\infty$ and there is a class $P_s$ such that $N(u)\cap P_s\neq N(v) \cap P_s$.

\item[(3)] Suppose that $u$ and $v$ are distinct vertices in the clique of $G[P_1]$ and $\sigma_0(u)=\sigma_0(v)=s<\infty$.

\item[(4)] Suppose that $u$ and $v$ are distinct vertices in the clique of $G[P_1]$ and $\sigma_0(u)=\sigma_0(v)=\infty$ and $N(u)\cap P_s\neq N(v) \cap P_s$ for some $s$.
\end{description}

In these cases the partition $(P_1-\{u,v\}, P_2, \ldots, P_s\cup \{u,v\}, P_{s+1}, \ldots, P_\ell)$ is a good partition, but the class $P_s\cup \{u,v\}$ induces a non-Eulerian subgraph.
By repeatedly applying Proposition \ref{kdcbasicfacts} we can split this class into classes of size at most 2. The new partition is a left-aligned maximal partition on more than $\ell$ classes. This algorithm makes $O(\ell n^2)$ steps as we can process $P_i$ for $i=2,\ldots, \ell$ (to decide whether there are two vertices in $P_1$ such that one of the rules can be applied with $s=i$) in time $O(n^2)$. If one of the rules can be applied, then we can construct a maximal left-aligned partition on at least $\ell +1$ classes in the same time window.

\begin{defn}[Equivalent vertices]
A connected vertex pair $v,w$ is called connected-equivalent, or in short c-equivalent if $N(v)-\{w\}=N(w)-\{v\}$. If $v$ and $w$ are disconnected and $N(v)=N(w)$, then they will be called disconnected-equivalent or d-equivalent.
\end{defn}

It is easy to verify that both c-equivalence and d-equivalence are equivalence relations.

\begin{lem} \label{thereareequivs}
Let $G$ be a graph and suppose $k\leq \maxdc(G)$. Let $\mathcal{P}$ be a maximal left-aligned partition of $G$ obtained by the \verb|Greedy algorithm| with  $\ell \leq k$ classes that cannot be improved by any of the above operations. If the first class of this partition has size at least $4k+1$, then it contains $k+2$ c-equivalent or $k+2$ d-equivalent vertices.
\end{lem}

\begin{proof}
There are at least $2k+1$ isolated vertices or a clique of size at least $2k+1$ in the graph spanned by the first class of the partition.

Suppose there are $2k+1$ isolated vertices. Since operation (1) is not possible, there are at most $\ell-1$ isolated vertices with $\sigma_1(v)<\infty$, so we have a set $W$ of at least $2k+1-(\ell-1)\geq k+2$ isolated vertices with $\sigma_1(w)=\infty$. The partition cannot be improved using operation (2), thus $N(v)\cap P_i$ is the same set for each $w \in W$ ($i = 2,3, \dots \ell$). Consequently the vertices of $W$ are d-equivalent.

If we have a clique of size $2k+1$ instead, we find a set $W'\subseteq P_1$, $|W'|=k+2$ of c-equivalent vertices similarly.
\end{proof}

\begin{remark}
Take a maximal left-aligned partition that cannot be improved using the above operations. For any pair of c-equivalent (or d-equivalent) vertices $u,v \in P_1$ we get that $\sigma_0(u)=\sigma_0(v)=\infty$ (or $\sigma_1(u)=\sigma_1(v)=\infty$).
\end{remark}

\begin{lem}\label{wecandelete}
Let $G$ be a graph and let $k\leq \maxdc(G)$. Let $\mathcal{P}$ be a maximal left-aligned partition of $G$ with  $\ell\leq k$ classes that cannot be improved by any of the above operations. If there is a set $W$ of $t\ge k+2$ c-equivalent or d-equivalent vertices in the first class $P_1 \in \mathcal{P}$, then for any $W'\subseteq W$, $0 \equiv |W'|\le t-k$, we have $\maxdc(G-W')=\maxdc(G)$. 
\end{lem}

\begin{proof}
Let $(R_1,R_2, \dots, R_{\maxdc(G)})$ be a left-aligned optimal partition of $G$. It is enough to prove the statement for $t=k+2$, for larger values of $t$ this argument can be repeated).  Observe that (for $i>1$) $R_i$ 
cannot contain a pair of d-equivalent or c-equivalent vertices, since $d(\cup_{j=1}^{i-1} R_j, R_i)$ 
 is odd. Thus $|W\cap R_1|\ge t-k+1\ge 3$, and as  we can swap equivalent vertices, we may suppose that $w_1, w_2\in W\cap R_1$. 
By deleting $w_1$ and $w_2$ we get a graph $G'$ for which the partition $(R_1-\{w_1,w_2\},R_2, \dots, R_{\maxdc(G)})$ is good, thus $\maxdc(G') \geq \maxdc(G)$. (We needed 3 equivalent points in the first class so that the class remains non-empty after the deletion.)

We remained 
 to show that $\maxdc(G')\leq \maxdc(G)$. For that take a left-aligned optimal partition of $G'$. If we put back the vertices $w_1,w_2$ into the first class, 
we get a good partition of $G$, since the parity of vertex degrees 
 is unchanged.
\end{proof}

Note that after the deletion of such a vertex set $W'$ from the first class of $\mathcal{P}$ we get a good $\ell$-partition of $G-W'$.

\begin{thm}
Given $G(V,E)$ and $k$ as input, there is an algorithm that either finds a good $k$-partition or it
yields 
a subset $K \subseteq V$ with the following properties:
\begin{enumerate}
\item[(1)] $|K| \leq 6k$
\item[(2)] $\maxdc(G[K]) \geq k\;$ if and only if $\;\maxdc(G) \geq k$
\item[(3)] For any good $k$-partition of $G[K]$ the partition obtained by adding the vertices of $V-K$ to the first class we get a good $k$-partition of $G$.
\end{enumerate}
The algorithm runs in $O(k^2n^2)$ time.
\end{thm}

\begin{proof}
We make an $\ell_0$-part maximal left-aligned partition using the \verb|Greedy algorithm|. After that we can improve the partition using the improvement operations on page \pageref{opspage}, resulting in a partiton of $\ell \geq \ell_0$ classes.
If $\ell\geq k$, then we can obtain a good $k$ partiton (by contracting the first few classes if needed). The algorithm writes this partition to the output.
If $\ell<k$, and $|V|>6k$, then the first partition has at least $6k-2(\ell-1)>4k+1$ vertices. By Lemma \ref{thereareequivs} we have a c- or d-equivalence class of size at least $k+2$. Lemma \ref{wecandelete} allows us to delete vertices from the first class, and thereby decrease its size under $4k+1$. The new vertex set $K\subsetneq V$ that we obtained has at most $6k$ vertices, and $\maxdc(G[K])=\maxdc(G)$.

Now we estimate the time needed for the computations. Finding a maximal partition using the \verb|Greedy algorithm| can be done in $O(kn^2)$ time. After this we make the possible improvements from page \pageref{opspage}. All improvements can be done in $O(k^2n^2)$ time (since we can do at most $k$ improvements, and each of them takes $O(kn^2)$ time). In the end we delete a set of vertices, which takes no more than $O(n^2)$ time.
\end{proof}

The graph spanned by $K$ is called the \emph{kernel} of the problem. For the kernel we can run Algorithm \verb|Recursive k-dual-critical|.

\begin{cor}
There is a fixed parameter tractable (FPT) algorithm, that given $G$ and $k$, computes a good $k$-partition, or concludes that none exists. The algorithm runs in $O(k^2n^2+(6k)^{2k})$ time. 
\end{cor}

\begin{cor}
Using the above algorithm repeatedly for $k=1,2,\dots, \maxdc(G)+1$ it is possible to compute $\maxdc(G)$ in $O(n^4+n^2(6\maxdc(G))^{2\maxdc(G)})$ time.
\end{cor} 

\end{document}